\documentclass{jocg}
\usepackage{amsmath,amsfonts}

\usepackage{scalefnt}
\usepackage{hyperref}
\usepackage{mathtools}
\usepackage{graphicx}
\usepackage[font=small,labelfont=bf]{caption}
\graphicspath{{figures/}}

\title{%
  \MakeUppercase{A universality theorem for allowable sequences with applications}%
}

\author{%
  Udo~Hoffmann%
  \thanks{\affil{Universit\'e libre de Bruxelles}, 
          \email{hoffmann.odu@gmail.com, keno.merckx@ulb.ac.be}}\,
  and Keno~Merckx\footnotemark[1]%
}

\usepackage{amsthm}
\theoremstyle{plain}
\newtheorem{theorem}{Theorem}

\newtheorem{corollary}[theorem]{Corollary}

\newtheorem{lemma}[theorem]{Lemma}
\newtheorem{proposition}[theorem]{Proposition}
\newtheorem{remark}[theorem]{Remark}
\theoremstyle{definition}

\newtheorem{observation}[theorem]{Observation}

\DeclareMathOperator{\conv}{conv}

\DeclareMathOperator{\sgn}{sgn}

\newcommand{\R}{\mathbb{R}}
\newcommand{\ER}{\ensuremath{\exists\mathbb{R}}}

\newcommand{\CCC}{\mathcal{C}{}}
\newcommand{\PPP}{\mathcal{P}{}}

\begin{document}
\maketitle

\begin{abstract}
  Order types are a well known abstraction of combinatorial properties of a
  point set. By Mn\"ev's universality theorem for each semi-algebraic set $V$ there
  is an order type with a realization space that is \emph{stably equivalent} to $V$.
  We consider realization spaces of \emph{allowable sequences}, a refinement of order types.
  We show that the realization spaces of allowable sequences
  are \emph{universal} and consequently deciding the realizability is complete in the \emph{existential theory of the reals} (\ER).
  This result  holds even if the realization space 
  of the order type induced by the allowable sequence is non-empty.
  Furthermore, we argue that our result is a useful tool for further geometric reductions.
  We support this by giving \ER-hardness proofs for the realizability of abstract convex geometries and for the recognition problem
  of visibility graphs of polygons with holes using the hardness result for allowable sequences. This solves two longstanding open problems.
\end{abstract}

\section{Introduction}
Combinatorial abstractions of properties of point sets are a useful tool in computational geometry. For example, many algorithms on point sets only require the \emph{order type} instead of the exact coordinates of the points, which makes many algorithms in computational geometry purely combinatorial.
We recall the definition of an \emph{(abstract) order type}.
A \emph{chirotope} of a point set $P$ in the plane is the mapping
 $\chi: P^{3} \rightarrow \{-1,0,1\}$, where 
\begin{align*}
\chi(p_1,p_2,p_3)=  \sgn\left( \begin{vmatrix}
   p_{1_x} & p_{1_y} & 1 \\
   p_{2_x} & p_{2_y} & 1 \\
   p_{3_x} & p_{3_y} & 1
\end{vmatrix}
\right).
\end{align*} 
The map $\chi$ encodes an orientation (clockwise, collinear, or counterclockwise) for triples of points.
For a set $E$, the pair $(E,\chi)$ is  an \emph{abstract order type}  if $\chi:E^{3} \rightarrow \{-1,0,1\}$ satisfies the \emph{rank 3 chirotope axioms}: for any $p_1,p_2,p_3,q_1,q_2,q_3$ in $E$,
\begin{enumerate}
\item $\chi$  is not identically zero.
\item $\chi(p_{\sigma(1)},p_{\sigma(2)},p_{\sigma(3)})=\sgn(\sigma) \chi(p_1,p_2,p_3)$ for $p_1,p_2,p_3 \in E$ and any permutation $\sigma$.
\item If  $\chi(q_1,p_1,p_2)\chi(p_1,q_2,q_3) \geq 0$ and $\chi(q_2,p_1,p_2)\chi(q_1,p_1,q_3) \geq 0$ and  $\chi(q_3,p_1,p_2)\chi(q_1,q_2,p_1) \geq 0$ then $\chi(p_1,p_2,p_3)\chi(q_1,q_2,q_3) \geq 0$.
\end{enumerate}
An abstract order type $(P,\chi)$  is \emph{realizable} if there exists a point set in the plane with order type $\chi$. The sign conditions given for three points translates geometrically into a clockwise ($\chi(p_1,p_2,p_3)$ negative) or counterclockwise ($\chi(p_1,p_2,p_3)$ positive) orientation of the points $p_1,p_2,p_3$. 
The concept of (abstract) order types appears under different names in the literature, for example
oriented matroids and CC-systems~\cite{knuth1992axioms} (counterclockwise-systems).

It is a natural question to ask how the point sets that agree with the combinatorial description -- the \emph{realization space} -- looks like.
For order types this question was answered by Mn\"ev~\cite{mnev1988universality} with his famous universality theorem:
For each \emph{primary semi-algebraic} set $S$ (a set described by strict polynomial inequalities and polynomial equalities) there is an order type whose realization
space is \emph{stably equivalent} to $S$. So the realization space can be as complex as possible.
From the computational point of view Mn\"ev's universality theorem implies that deciding if an
order type is \emph{realizable} (i.e., the realization space is non-empty) is complete in the
\emph{existential theory of the reals} (\ER), thus at least NP-hard.

In this paper we consider another combinatorial description of a point set, the \emph{allowable sequence} or \emph{circular sequence}. We use $[n]$ to denote $\{1,\ldots,n\}$.
An allowable sequence is a sequence of permutations $\pi_1,\dots,\pi_k$ of the elements of $[n]$
such that every pair of elements is reversed exactly once in the sequence, and two consecutive permutations
differ by a \emph{move}, the reversal of disjoint substrings of the permutation.
An allowable sequence can be obtained from a point set in the following way.
The orthogonal projection of the point set onto an oriented line leads to a permutation.
Rotating this line by a continuous motion by $180^\circ$ leads to a sequence of permutations which is an allowable sequence.

\begin{figure}[ht]
  \centering
  \includegraphics[width=.5\textwidth]{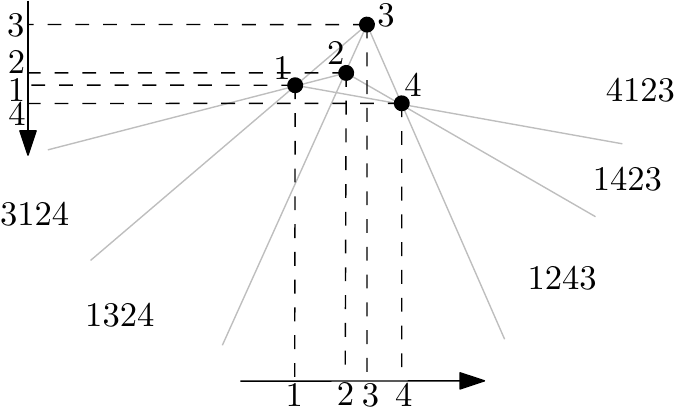}
  \caption{\label{fig:seq} An example of an allowable sequence of a point set.}
\end{figure}

We show that the same type of universality as for order types also holds for allowable sequences. 
While this is not a surprising result since allowable sequences are a refinement of order types
(i.e., the allowable sequence determines the order type~\cite{goodman1980combinatorial}),
we show that this result holds even if the order type determined
by the allowable sequence is realizable.
We generalize the \ER-completeness result to \emph{simple} allowable sequences,
i.e., sequences where each move between two permutations is a \emph{swap} of two adjacent elements.
Realizable simple allowable sequences correspond to point sets in general position with the additional condition that no two lines spanned by different pairs of points are parallel.

Furthermore, we argue that the \ER-hardness of allowable sequence realizability can be a very useful tool
for further reductions.
To support this we show \ER-hardness for two other longstanding open problems
-- the realizability of convex geometries and the recognition of vertex visibility graphs of polygons with holes.

In 1985, Edelman and Jamison~\cite{Edelman_Jamison_1985} gave an axiomatic abstraction of finite convex geometries.
A pair $(V,\CCC)$ with $\CCC\subseteq 2^V$ is called a convex geometry if
\begin{enumerate}
	\item $\emptyset \in \CCC$ and $V\in\CCC$,
	\item the set $\CCC$ is closed under intersection, and
	\item for each $C\in\CCC\setminus V$ there
	is $x\in V\setminus C$ such that $C\cup \{x\}\in \CCC$. 
\end{enumerate}
Starting with a set of points $P$ in $\R^{d}$,
one can obtain an abstract convex geometry $(P,\CCC_P)$ by defining $\CCC_P$ as $\{C\subseteq 2^{\PPP}: \conv(C) \cap P = C\}$
where the operator $\conv(C)$ is the classic convex hull of the set $C$.
For a convex geometry $(V,\CCC)$ we say that $(V,\CCC)$ is realized by a set of points $P$ in $\R^d$ if  $(V,\CCC)$ is isomorphic to $(P,\CCC_P)$.
Edelman and Jamison posed the problem of characterizing the convex geometries that can be realized as convex geometries of a finite point set in $\R^d$.  

For $d=2$, Adaricheva and White~\cite{Adaricheva10a} have shown that the convex geometry determines
the order type of the realizing point set once the cyclic order of the points on the convex hull is
assumed to be fixed, showing that the algorithmic problem of deciding a restricted version of the \emph{Edelman--Jamison problem} is \ER-hard.
We give a simple construction that shows \ER-hardness of deciding the realizability of convex geometries in $\R^2$ from
using the fact that deciding the realizability of allowable sequences is \ER-hard.
This construction can be lifted in any dimension, thus solves the Edelman--Jamison problem completely.

The \emph{(vertex) visibility graph of polygons} is defined as a graph $G=(V,E)$ where $V$ is the set
of vertices of the polygon and a pair of vertices is in $E$ if the two vertices can \emph{see} each other, i.e.,
the segment spanned by the vertices is entirely contained in the polygon.
The recognition problem of polygon visibility graphs, i.e., given a graph $G$, decide whether $G$ is the
visibility graph of a polygon, has been studied extensively.
However, no hardness results nor a (non-deterministic) polynomial-time algorithm for this problem are known.
Abello and Kumar~\cite{AK02} pointed out connections to \emph{oriented matroid realizability},
which is an equivalent problem to order type realizability.
We consider the visibility graph of polygons with holes or \emph{polygonal domains}.
Using allowable sequence realizability, we show that the recognition of visibility graphs
of polygons with holes is \ER-complete.

\subsection{Related work and Connections}\label{subsec:rel}

The \emph{existential theory of the reals} (\ER) is a complexity class defined by the following complete problem:
given Boolean combination of polynomial equalities and inequalities, decide if there is an assignment of
real values to the variables such that the system is satisfied.
The only known relations to other complexity classes are $NP\subseteq\ER\subseteq PSPACE$.
The second relation is a result by Canny~\cite{canny1988}.
Whether one of the relations is strict or an equality is not known.
An argument suggesting that \ER~may not be contained in NP is the fact that the
natural certificate for \ER-hard problems, the coordinates of a solution, cannot
always be stored in polynomial space since iterative squaring produces doubly exponential numbers in the
input size, which requires exponential size binary representations~\cite{MM13}.

The complexity class \ER\ has been introduced Schaefer~\cite{S09}, motivated by the
continuously expanding
list of (geometric) problems that have been shown to be complete in this class.
One of the first geometric \ER-complete problems is \emph{order type realizability}:
Given an orientation (clockwise, counterclockwise, collinear) of each triple of elements of a ground set $P$,
decide if $P$ can be mapped to points in the plane such that the triples of points have the prescribed orientation.
The \ER-hardness of order type realizability has been shown by Mn\"ev with his famous
universality theorem:
For each polynomial inequality system $S$ that consists of strict inequalities
and equalities (\emph{primary semi-algebraic set})
there is an order type whose \emph{realization space} is \emph{stably equivalent} to $S$.
This means that the order type constructed from $S$ is not only realizable if the set solutions of $S$
is non-empty, but the set of solutions also has the same ``structure''.
The structures preserved by the \emph{stable equivalence} relation include the homotopy group and the algebraic
complexity of the numbers (e.g., ``Does $S$ have a rational solution''),
see~\cite{richter1995realization} for more consequences of stable equivalence.
We give a more detailed introduction to \emph{stable equivalence} in Section~\ref{sec:allow}.

Many geometric problems can be shown to be \ER-hard by a reduction from order type realizability
or the dual problem, the stretchability of  pseudoline arrangements.
Some \ER-complete geometric problems include recognition of segment~\cite{kratochvil1994intersection},
disc~\cite{MM13}, and convex set intersection graphs~\cite{S09}
and point visibility graphs~\cite{cardinalPVG}, the art gallery problem~\cite{abrahamsen2017art}, realizability of the face lattice of a $4$-polytope~\cite{richter1995realization} and
$d$-dimensional Delaunay triangulations~\cite{APT15}, computing the rectilinear crossing number~\cite{B91} and planar slope number~\cite{planarSlopes}.
We refer to~\cite{cardinal2015computational} for an overview.


The concept of an allowable sequence has first been described by Perrin~\cite{perrin1882probleme}
in 1882 who conjectured that every allowable sequence is realizable.
This was disproved by Goodman and Pollack~\cite{goodman1980combinatorial} by showing that the ``bad pentagon''
-- an allowable sequence that \emph{induces} the order type of a convex $5$-gon -- is not realizable.
This paper also introduces the term ``allowable sequence'' for point sets in general position.
The authors generalized the definition in~\cite{goodman1982}
to what we call \emph{generalized allowable sequence} in this paper.
Allowable sequences have several applications in combinatorial geometry, see for example~\cite{goodman1993circular}. 
The computational complexity of the realizability problem has been posed as an open question in~\cite{pilzDiss}.
A slightly different version of the \ER-hardness result can be found
in the first authors PhD thesis~\cite{hoffmann2016thesis}.


In their paper Edelman and Jamison~\cite{Edelman_Jamison_1985} have developed the foundations of a combinatorial abstraction of convexity.
Similar ideas were studied by Dilworth~\cite{Dilworth_1940} and later by Korte, Lov{\'a}sz and Schrader~\cite{Korte_Lovasz_Schrader_1991}
via the notion of \emph{antimatroid}, a concept dual to the one of a finite
convex geometry.
Today, the concept of convex geometry appears in many fields of mathematics such as formal language theory~\cite{Boyd_Faigle_90},
choice theory~\cite{Koshevoy_1999} and mathematical psychology~\cite{Falmagne_Doignon_LS} among others.
Kashiwabara et al.~\cite{kashiwabara2005affine} showed that any abstract convex geometry can be represented as ``generalized shelling'' in $\R^d$ for some $d$.
Richters and Rogers~\cite{richter2015embedding} reproved this theorem giving a
better bound on the dimension.
Different representations of finite convex geometries using different shapes than points for
the ground set have been
studied~\cite{ada16circles,czedli16almostCircles,czedli2014finite}.

Visibility graphs have been an active field of research over the last 40 years.
Polygon visibility graphs have many applications, for example motion planning in robotics~\cite{LP79}.
So there have been some attempts to give a combinatorial characterization of point visibility graphs.
Sequences of papers~\cite{G88,G97} proposed necessary conditions and conjectured them to be sufficient, which has
always been disproved~\cite{S05}.
The problems here are due to stretchability/realizability issues, which motivated the notion of pseudo-visibility graphs~\cite{ORS97}
(which have a similar relation to polygon visibility graphs as pseudoline arrangements to line arrangements).
Pseudo-visibility graphs can be characterized combinatorially~\cite{gibson2015characterization}.
All those characterizations use the Hamiltonian cycle of the visibility graph induced by the outer boundary of a hole-free polygon.
Even if this Hamiltonian cycle is given, the complexity of the recognition problem for (pseudo-)visibility graphs is open.
The only hardness result is the NP-hardness of a ``sandwich version'' of the recognition problem of polygon visibility graphs~\cite{chaplick2016}, i.e., deciding if for two input graphs $G\subseteq H$ there exists a polygon visibility graph $K$ with $G\subseteq K\subseteq H$ is NP-hard.
There are a couple positive results characterizing and recognizing visibility graphs of special classes of polygons,
for example for \emph{spiral polygons}~\cite{everett1990recognizing}, \emph{tower polygons}~\cite{colley1997visibility}, and \emph{anchor polygons}~\cite{boomari2016}.
To our knowledge, the only result involving the characterization of visibility graphs of polygons with holes is for \emph{polygonal rings},
visibility graphs of a convex polygon with one convex hole~\cite{cai1995polygonalrings}.

\subsection{Outline of the paper}
In Section~\ref{sec:allow} we show the \ER-completeness of the realizability problem for allowable sequences.
We proceed by applying this result to show that the realizability problem for abstract convex geometries is \ER-complete in Section~\ref{sec:cg}.
Finally, in Section~\ref{sec:pv} we show that the recognition problem for polygon visibility
graphs with holes is \ER-complete.

We only show \ER-hardness proofs in this paper. For \ER-membership proofs of
the problems we discuss we refer to~\cite{everettthesis} for polygon visibility
graphs. The complexity of the recognition problem for polygon visibility graphs without holes remains open.

\section{Allowable sequences}\label{sec:allow}
In this section we prove the universality theorem for generalized allowable sequences.
We first give some basic properties and definitions.
Afterwards we prove the universality theorem for generalized allowable sequences in Subsection~\ref{subsec:universal}
and show in Subsection~\ref{subsec:simple} how this theorem also implies \ER-hardness for simple allowable sequences.

\subsection{Basics and notations}\label{sub:basicallow}
One basic property describing the relation between allowable sequences and order type is given by the following lemma.
\begin{lemma}[\cite{goodman1980combinatorial}]
  A (generalized) allowable sequence determines an abstract order type.
\end{lemma}
In other words, allowable sequences carry more information than the order type relation.
When we mention this order type connected to the allowable sequence we refer to it as the order type \emph{induced} by the allowable sequence.


A \emph{generalized allowable sequence} is a sequence of permutations $\pi_1,\dots,\pi_k$ of the elements of $[n]$
such that every pair of elements is reversed exactly once in the sequence, and two consecutive permutations
differ by the reversal of several (non-overlapping) substrings. We call each reversal of a single substring a \emph{move}. The elements of one substring $s$ that is reversed correspond to points that lie on one line $\ell_s$. This substring is reversed when the rotating line, which we project on to obtain the permutations, is orthogonal to $\ell_s$. We can identify a switch with its intersection point on $\ell_\infty$, the line at infinity, and a single permutation with<
the interval between two intersection points on the line at infinity.
Thus it is not surprising that the allowable sequence carries the same information as the ordered sequence of switches.
When talking about allowable sequences we will often work in the projective plane and treat the switches as the points on the line at infinity.

In a generalized allowable sequence we can have \emph{parallel switches} where two switches $s_1$ and $s_2$ appear in one move.
Those correspond to switches where $\ell_{s_1}$ and $\ell_{s_2}$ intersect on the point on $\ell_\infty$,
in other words: $\ell_{s_1}$ and $\ell_{s_2}$ are parallel.
In a non-generalized allowable sequence the order of the switches is the order of the intersection points
of the spanned lines with a large enough cycle that contains all the intersection
points of spanned lines in its interior.
The order of switches corresponds exactly to the order of slopes spanned by the line of the switch.

For the universality theorem we need the definition of \emph{stable equivalence}. We use the definition of~\cite{richter1995realization}.

Stable equivalence is an equivalence relation induced by \emph{rational transformations} (i.e., the function and its inverse are rational functions) and \emph{stable projections}.

A semi-algebraic set $V$ is a stable projection of $W$ if $V$ is obtained from $W$ by the projection on the first $n$ coordinates and
\[W=\{ v \oplus v' \mid v\in V\mbox{ and } \phi_i(v)^T\cdot v'>0 \mbox{ and }\psi_j(v)^T\cdot v'=0\mbox{ for }i\in X,j\in Y\},\]
where $X$ and $Y$ are index sets and $\phi_i$ and $\psi_j$ are polynomial functions from $\R^n\rightarrow \R^d$ (i.e., the functions are defined by a polynomial in each of the $d$ coordinates).
In other words, for each $v\in V$ we obtain the preimage of $v$ under a stable projection by concatenating $v$ with a polyhedron defined by polynomial functions of $V$.
Two semi-algebraic sets $V$ and $W$ are \emph{stably equivalent} if they lie in the same equivalence class
induced by rational transformations (which include projective transformations) and stable projections.

Note that the set of $v'$ for one fixed $v$, which we project away by a stable projection comes from a polyhedron, which is a contractible set.
This shows that the homotopy type of $W$ is preserved under stable projections.
For more invariants of semi-algebraic sets under stable equivalence we refer to~\cite{richter1995realization}.

\subsection{Universality for generalized allowable sequences}\label{subsec:universal}
In this subsection we give the idea of the proof of Mn\"ev's universality
theorem~\cite{mnev1988universality} based on Shor's proof~\cite{shor1991stretchability}.
We modify this proof to show the following theorem.

\begin{theorem}\label{thm:nonSimple}
  For every primary semi-algebraic set $V$ there is a generalized allowable
  sequence $A$ whose
  realization space is stably equivalent to $V$.
\end{theorem}

Therefore, we decompose the description of the primary
semi-algebraic set into elementary arithmetic operations on three
variables. This can be done using the following theorem.

\begin{theorem}[\cite{shor1991stretchability}]
  Every \emph{primary} semi-algebraic set $V\subseteq\mathbb{R}^d$ is \emph{stably equivalent}
  to a semi-algebraic set $V'\subseteq\mathbb{R}^n$, with $n=\mathrm{poly}(d)$, 
  for which all defining equations have the form 
  $x_i+x_j=x_k$ or $x_i\cdot x_j=x_k$ for certain $1 \leq i \leq j < k\leq n$, where the
  variables $1=x_1<x_2<\dots<x_n$ are totally ordered.
\end{theorem}

In addition to the decomposition of a polynomial equation into single addition or multiplication
steps, the theorem also gives us a total order of the size of the
variables in any representation. We call the description of the
semi-algebraic set the \emph{normal form}.

The elementary calculations in the normal form can implemented by order
types using the classic idea of von Staudt sequences~\cite{staudt}.
Since distances are not invariant under projective transformations we
use a projective invariant, the \emph{cross-ratio}.

The cross-ratio $(a,b;c,d)$ of four points $a,b,c,d\in\mathbb{R}^2$ is defined as
$$(a,b;c,d):=\frac{|a,c|\cdot|b,d|}{|a,d|\cdot|b,c|},$$
where $|x,y|$ is the determinant of the matrix obtained by writing the two vectors as columns.
The two properties that are useful for our purpose is that the cross-ratio
is invariant under projective transformations, and that for four points on one line, the cross-ratio
is given by $\frac{\overrightarrow{ac}\cdot\overrightarrow{bd}}{\overrightarrow{ad}\cdot\overrightarrow{bc}}$,
where $\overrightarrow{xy}$ denotes the oriented distance between $x$ and $y$ on the line.

The gadgets forcing a certain cross-ratio of points on a line $\ell$ are depicted in Figure~\ref{fig:gadgets}.
\begin{figure}[ht]
  \centering
  \includegraphics[width=.8\textwidth]{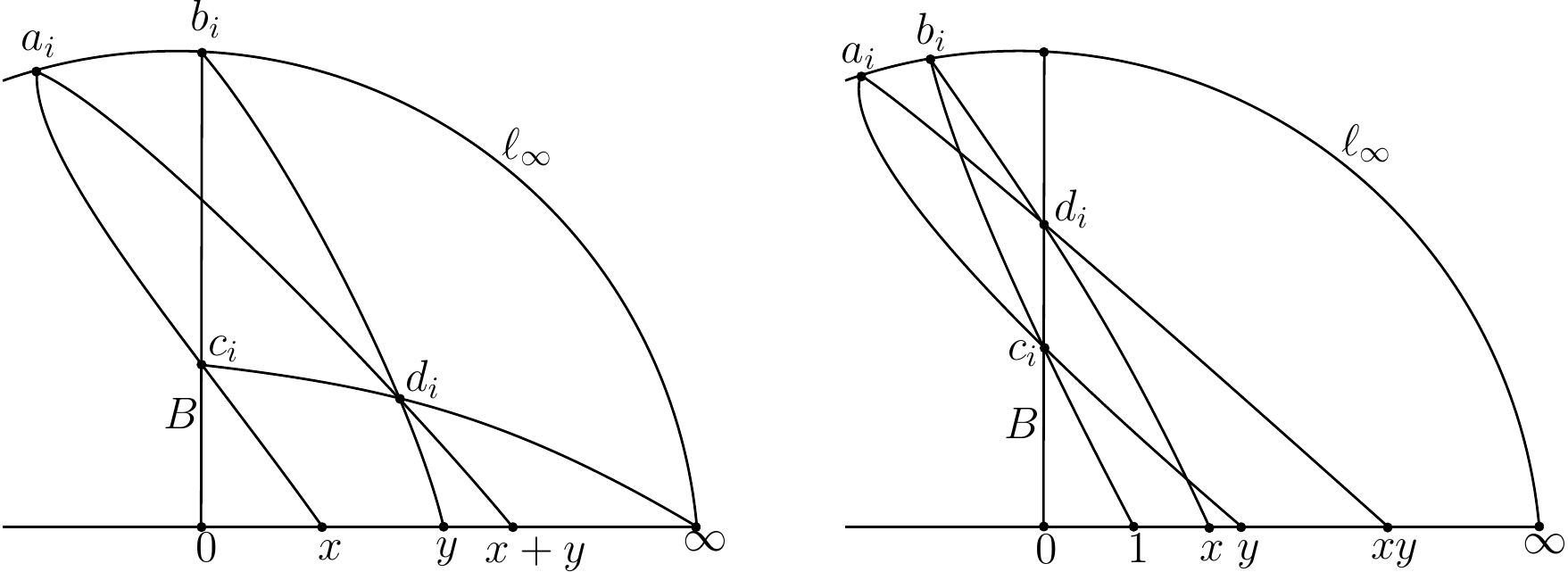}
  \caption{\label{fig:gadgets} Gadgets for addition (left) and multiplication (right) on a line.}
\end{figure}

We can construct an order type from the normal form in the following way.
For each variable in the normal form we place a point on a horizontal line $\ell$.
The order of those points on $\ell$ respects the total order of the variables.
We have some freedom where we place the gadget points.
This freedom allows to place the gadgets, so that we can determine a complete order type that can be realized if and only if the semi-algebraic set is non-empty.

We use the freedom we have in positioning the gadgets for the underlying order type of the allowable 
sequence and place the variables on the line at infinity.
Using a projective transformation switch the positions of $\ell$ and $\ell_\infty$, i.e., we place $\ell$ on the line at
infinity and $\ell_\infty$ on the $x$-axis.
\begin{figure}[ht]
  \centering
  \includegraphics[width=.8\textwidth]{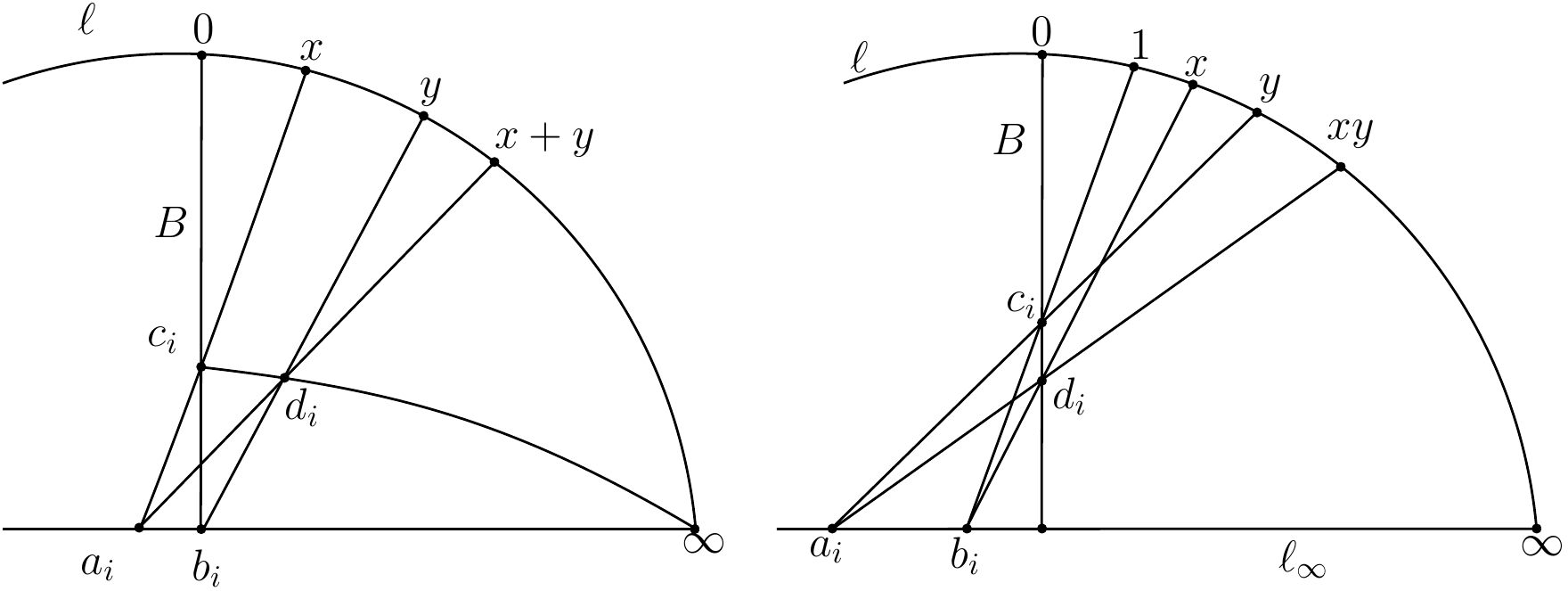}
  \caption{\label{fig:gadgets_reversed} The gadgets for addition (left) and multiplication (right)
    with $\ell$ as line at infinity.}
\end{figure}
To define the allowable sequence we show how we realize the point set if $V$ is non-empty,
where the allowable sequence of the points we construct is independent of the value determined by the point in $V$.
The combinatorial allowable sequence can be obtained from the construction, without actually having a representation. 
This allowable sequence is then realizable if and only if $V$ is non-empty since it realizes all the gadgets which determine the calculations which implement $V$.

We start with the construction of the allowable sequence.
First, we enumerate the gadgets $g_1,\dots,g_k$ such that the addition gadgets are ordered according to decreasing value of the involved $y$-variable of the gadget. Futhermore, we assume the addition gadgets have a smaller index than the multiplication gadgets.
These two conditions are helpful in Lemma~\ref{lem:realizable}, when we construct a realization of the order type that is induced by the allowable sequence.
Then we place the points of the gadgets iteratively starting with the first.
\begin{figure}[ht]
  \centering
  \includegraphics[width=.8\textwidth]{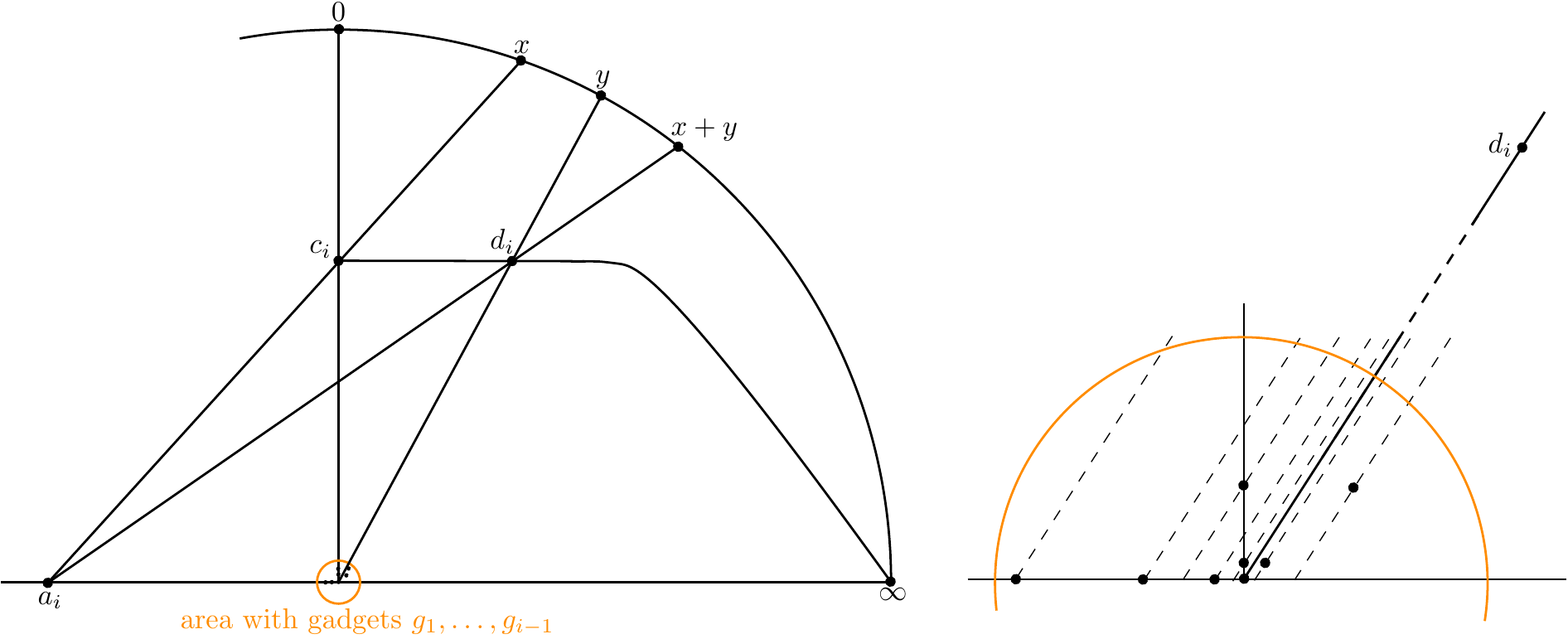}
  \caption{\label{fig:placing} Left: Placing the (addition) gadget $g_i$ after the gadgets $g_{1},\dots,g_{i-1}$ are already placed. Right: The projection of points in direction $d_i$.}
\end{figure}

We assume the gadgets $g_1,\dots,g_{i-1}$ are already placed.
If $g_i$ is an addition gadget we place the point $a_i$ on the
coordinate $(-N,0)$ for some large $N$. The point $c_i$ will then lie
on $(0,M)$ for a large $D$, and the point $d_i$ on $(1/y\cdot D,D)$. If we choose $N$ (and thus $M$) large enough,
then all points of the gadgets $g_1,\dots,g_{i-1}$ lie relatively
close to the origin compared to the points of $g_i$. Thus the lines
through $a_i$ and a point of $g_k$ for $k<i$ are almost
horizontal, thus we can determine the position of all switches
including $a_i$.
Similarly, the lines through $c_i$ are almost vertical and the
lines through $d_i$ and a point of $g_k$ have almost the slope
 of the line through the origin and the
point $y$ on $\ell$ (the line at infinity). 
We know the relative position of all
those switches since we know the complete allowable sequence of the
gadget points placed so far. Figure~\ref{fig:placing} right shows this projection in direction $d_i$.

For the multiplication gadget we apply the same construction. By placing $a_i$ and $b_i$ at coordinates $(-N,0)$ and $(-N-\varepsilon,0)$ for a sufficiently large $N$ and a small positive $\varepsilon$ we have determined the position of $c_i$ and $d_i$.
With $N$ large enough we again know the complete allowable sequence:
the switches including $a_i$ and $b_i$ and the points of the already placed gadgets are almost horizontal and can be determined using the horizontal projection of the gadget points.
Similarly, the switches including $c_i$ or $d_i$ and previous gadget points can be determined using the vertical projection.

\begin{proof}[Proof of Theorem~\ref{thm:nonSimple}]
  We first show that we obtain a realization of $A$ that corresponds to a point in the semi-algebraic set $V$.

  In the construction described above, we have constructed an allowable sequence $A$
  from a given semi-algebraic set $V$. The construction shows that from each point in $V$
  (i.e., choice of coordinates on $\ell$) there exists a realization of $A$ with these coordinates.
  This shows that each point in $V$ can be obtained by a projection of one realization of $A$.
  On the other hand, each realization of $A$ leads, after a projection and suitable affine transformation
  ($x_0=0$ lies on $(0,0)$ and $x_1=1$ lies on $(1,0)$) to a point of $V$ by considering
  the first coordinate of the points on $\ell$.
  
  This already shows that $A$ is realizable if and only if $W$ is non-empty.
     To show stable equivalence of the realization space of $A$ we first have to say how exactly we define the realization space of $A$.
    Therefore, we consider not only the points $P$ that realize $A$ but also the switch points on the line at infinity.
    In a first step we use a linear transformation (which leaves the allowable sequence invariant)
    to place the points of $P$ on $\ell_\infty$ on the line given by $y=0$,
    and the line $B$ between the variable $0$ and the point $b_i$ of an addition gadget to the line $x=0$.
    
    In a second step we iteratively project away the points of the last gadget and the points of the
    allowable sequence spanned by those points using stable projections.
    Note that this boils down to determining an interval for the point $a_i$ (and $b_i$ in a multiplication gadget) on the line $\ell_\infty$,
    such that all lines spanned points of the gadgets lie in the correct interval. The remaining points of the gadget are uniquely determined by the points $a_i$ and $b_i$. 
    
    For the proof of universality we still have to show that we can describe the
    position of the points of one gadget and the points on the line at infinity using (inequalities) only using the inner products.
    
    Let $a=(a_x,a_y)$ be the vector spanning the line $\ell_a$ and $b$ be the vector spanning the line $\ell_b$, both in positive $x$-direction.
    Then the inner product $(a_y,-a_x)\cdot b$ is the inner product of $b$ with a vector orthogonal to $a$ as shown in Figure~\ref{fig:slope_comparison}.
    This product is larger than $0$ if the slope of $\ell_b$ is smaller than the slope of $\ell_a$,
    equal to $0$ if the slopes are the same, and smaller than $0$ if the slope of $b$ is larger.  
    This follows from the fact that the inner product gives the oriented length of the projection of one vector onto the other one.
    The vectors spanning the lines are calculated as the difference of two points.
    \begin{figure}[ht]
    	\centering
    	\includegraphics[width=.5\textwidth]{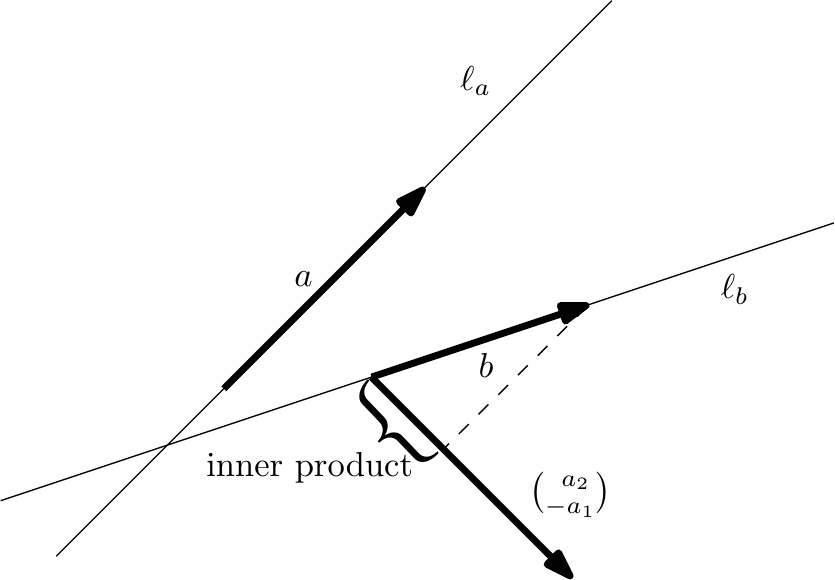}
    	\caption{\label{fig:slope_comparison} The slope of $\ell_a$ is larger than the slope of $\ell_b$, thus the scalar product is positive.}
    \end{figure}
    Thus we can compare slopes of two lines using the inner product.
\end{proof}

\begin{lemma}\label{lem:realizable}
  The order type induced by the allowable sequence $A$ in the proof above is realizable, even if $A$ is not realizable. 
\end{lemma}
\begin{proof}
We construct a realization of the order type $O$ induced by the
 allowable sequence $A$.
Therefore, we first give a short description of the order type.
We start with the addition gadgets as shown in Figure~\ref{fig:realizable}.

\begin{figure}[ht]
	\centering
	\includegraphics[width=.95\textwidth]{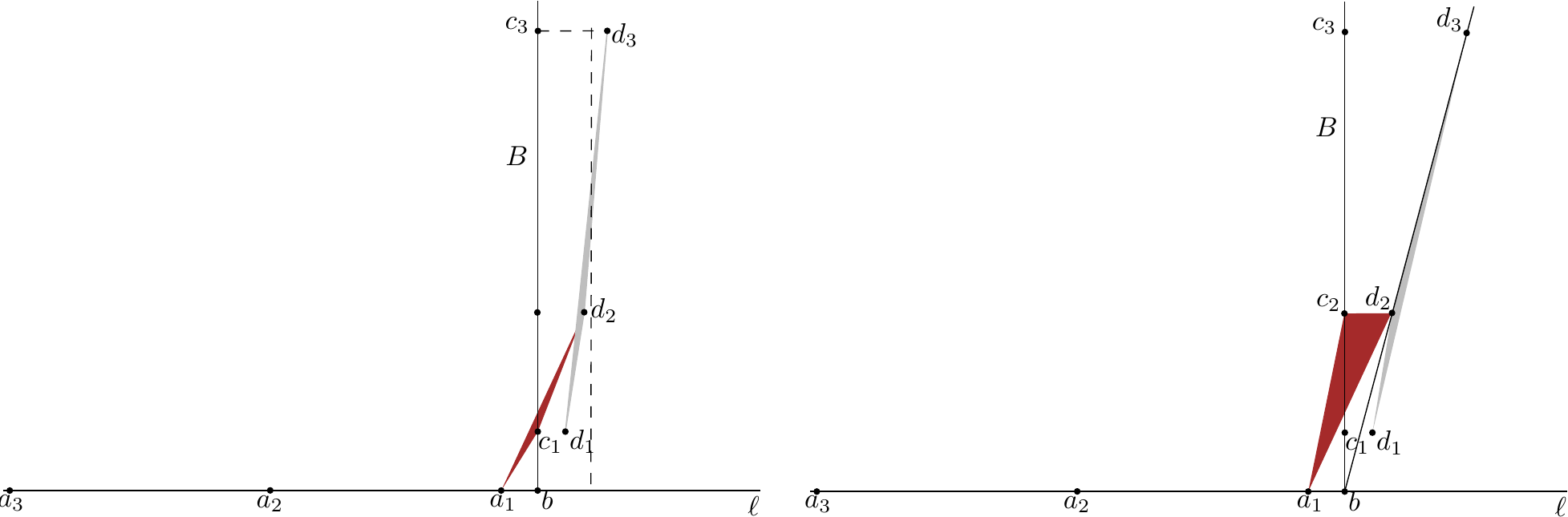}
	\caption{\label{fig:realizable} Left: The order type using only the addition gadgets when no two gadgets use the same $y$-variable. Right: Gadget $g_2$ and $g_3$ use the same $y$-variable.}
\end{figure}

If the addition gadgets do not use the same $y$-variable. In this case the points $d_i$ of the gadgets form a convex chain, such that the points $c_i$ and $d_i$ are in non-strict convex position as indicated by the thin grey triangle in Figure~\ref{fig:realizable} left.
A triple $(a_i,d_j,c_k)$ is oriented clockwise if $k<j$ and counterclockwise otherwise, see the brown triangles in the figure.
If two gadgets use the same $y$-variable as in Figure~\ref{fig:realizable} (right), then the points $d_i$ are not in convex position any more. Removing all but one of the points using the same $y$ leads to an order type where the $d_i$ are in convex position. The orientation of $(d_i,d_j,d_k)$ ($i<j$), where $d_i$ and $d_j$ belong to gadgets with the same $y$-variable is clockwise if $k<i$ and counterclockwise otherwise, see the grey triangle in the right figure.

We realize this order type by iteratively placing the gadgets. We assume the first $i-1$ gadgets are placed. Then we place $a_i$ to the left of every intersection point of the line $\ell$ (that supports all the $a$ points of the already placed gadgets) with the lines spanned by the already placed gadget points. The point $c_i$ is placed in the same way above the highest intersection point of a spanned line on $B$. The point $d_i$ is placed just to the right of $d_{i-1}$ on the same $y$-coordinate of $c_i$ as indicated for $d_3$ by the dashed line in Figure~\ref{fig:realizable} left.
Placing $d_i$ close enough to this vertical line leads to a realization that has the correct orientation for the grey triangles. In the case that gadget $g_i$ uses the same $y$-variable as $g_{i-1}$ 
we place $d_i$ on the line spanned by $b$ and $d_{i-1}$ and on the horizontal line through $c_i$.
This has the effect that the grey triangle has the correct orientation.

Afterwards we place the multiplication gadgets, the points $a_i$ and $b_i$ are placed to the left of all already placed points on $\ell$ and $c_i$ and $d_i$ on $B$ above every point. The order type of those points with the points on $B$ and $\ell$ is already correct by construction. The orientation including a point of a multiplication gadget point $m$ on $\ell$ $(m,n,o)$ is clockwise if the $y$-coordinate of $n$ is larger than the $y$-coordinate of $o$, see the brown triangle in Figure~\ref{fig:realizable_mult}.
\begin{figure}[ht]
	\centering
	\includegraphics[width=.95\textwidth]{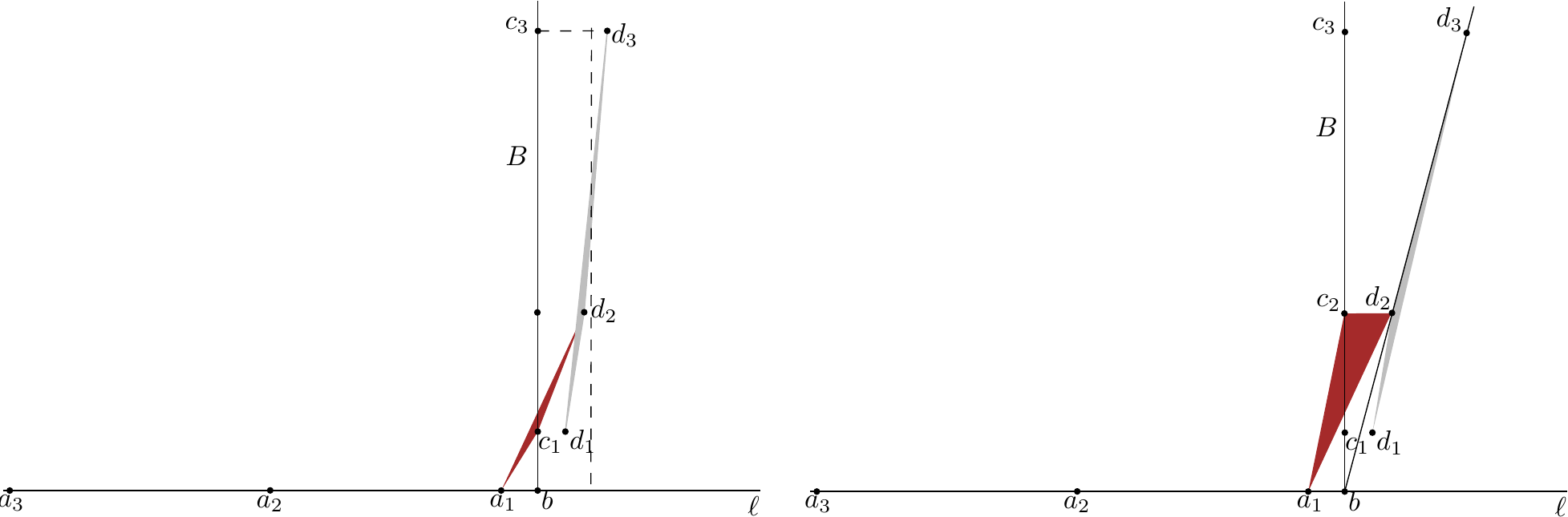}
	\caption{\label{fig:realizable_mult} Left: The order type using only the addition gadgets when no two gadgets use the same $y$-variable. Right: Gadget $g_2$ and $g_3$ use the same $y$-variable.}
\end{figure}
In the case of an equal $y$-coordinate (,i.e, $n$ and $o$ are $c_i$ and $d_i$ of an addition gadget,) $(m,c_i,d_i)$ is oriented clockwise.
This shows that the order type induced by the allowable sequence in the proof above is realizable, even if the allowable sequence is not realizable.
\end{proof}
All proofs in this section are constructive and can be turned into polynomial time algorithms, which results in the following corollary.
\begin{corollary}
  The realizability problem for generalized allowable sequences is \ER-complete.
\end{corollary}

\subsection{Simple allowable sequences.}\label{subsec:simple}
In this subsection we generalize Theorem~\ref{thm:nonSimple} to simple allowable sequences, which is a useful tool to show the \ER-hardness of other geometric realizability problems.

\begin{theorem}\label{thm:simple}
  The realizability of simple allowable sequences is \ER-complete.
\end{theorem}

To prove the theorem above we extend the method of \emph{constructible} order types,
that has already been used to show that simple order type realizability is \ER-complete~\cite{orientedMatroids,shor1991stretchability,matousek2014segment}.

An order type is \emph{constructible} if there is an order of the points $(p_1,\dots,p_n)$
such that
\begin{itemize}
\item the points $p_1,\dots,p_4$ are in general position,
\item the point $p_i$ lies on at most two lines spanned by the points $p_1,\dots,p_{i-1}$.
\end{itemize}

\begin{lemma}[\cite{shor1991stretchability}]
  For a constructible order type $O$ there is simple order type $O'$ that is realizable if and only if $O$ is realizable.
\end{lemma}

\begin{figure}[ht]
  \centering
  \includegraphics[width=.25\textwidth]{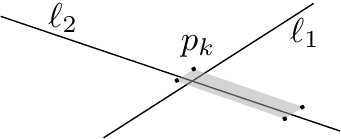}
  \caption{\label{fig:simple_replace} Replacing a point $p_k$ by four points.}
\end{figure}
The idea to construct the order type $O'$ from $O$ is the following.
Let $p_k$ the last point in the order that lies on at least one line spanned by the points $\{p_1,\dots,p_{k-1}\}$.
We assume $p_k$ actually lies on two lines $\ell_1$ and $\ell_2$. This is the more difficult case (compared to a point on one line).
We replace $p_k$ by four points in convex position, such that the order type encodes that $\ell_1$ and $\ell_2$ intersect in the convex hull of the new points.
The new order type $O'$ is realizable if $O$ is realizable since we can also do this replacement in a realization.
On the other hand, we can also put $p_k$ back into a realization of $O'$, namely at the intersection point of $\ell_1$ and $\ell_2$.
The four new points make sure that this intersection point has the same orientations as $p_k$, which guarantees that $O'$ is realizable if and only if $O$ is realizable.\footnote{We need that $p_k$ only lies on two lines since this method only ensures that all pairwise intersection points of the lines lie in the convex hull of the new points, but only for two lines we know that the pairwise intersection point is one point.}
The order type $O'$ is again constructible and all collinearities appear among the first $k-1$ points. Thus iterating this process leads to a simple order type that is realizable if and only if $O$ is realizable.
We also show how to adapt this replacement procedure 
to allowable sequences to obtain a simple allowable sequence $A'$ that is realizable if 
and only if the allowable sequence $A$ from Theorem~\ref{thm:nonSimple} is realizable.

\begin{proposition}[\cite{shor1991stretchability,matousek2014segment}]\label{obs:construct}
	The order types in Mn\"ev's universality theorem are constructible.
\end{proposition}

\begin{proof}
	We give a constructability sequence of the order type above.
	The first points $p_1,\dots,p_3$ are $0,\infty$ and $M$
	and a point $A$ that lies on a line between $\infty$ and $1$. Note we need only 3 points in the constructible part.
	Now we need a property of the normal form of the primary semi-algebraic set.
	Namely, that each variable is the output variable of at most one gadget,
	i.e., each variable appears in at most one equation of the form
	$x_i+x_j=x_k$ or $x_i\cdot x_j=x_k$ as $x_k$.
	This property allows an order of the gadgets, such that the output
	variable of the gadget is not constructed yet.
	With this observation we can assume the input variables of a
	gadget already exist and construct the gadget points
	in the order shown in Figure~\ref{fig:construct}.
\end{proof}
\begin{figure}[ht]
	\centering
	\includegraphics[scale=1]{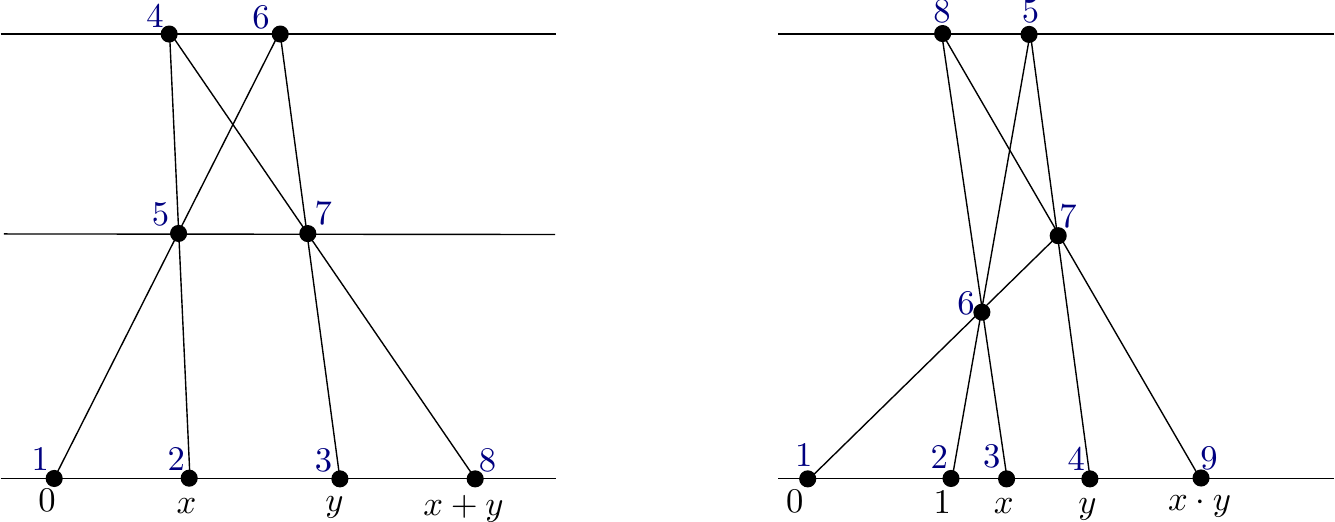}
	\caption{\label{fig:construct} The constructability sequence of points
		within one gadget.}
\end{figure}

To adapt this replacement procedure for allowable sequences we use the following idea.
We construct an order type $O_A$ that consists of the points $P$ and additional points $\{o_1,\dots,o_k\}=:O$ of $A$ on the line at
infinity that represent the switches of the allowable sequence\footnote{We can determine this order type.}.
For an allowable sequence $A$ constructed in the proof of Theorem~\ref{thm:simple} this order type $O_A$
is essentially the order type of Shor's proof for Mn\"ev's universality theorem with additional points on $\ell$ (the line at infinity)
that represent the switches that do not correspond to variables.
Those additional points do not obstruct the constructability of the order type $O_A$, since they
come from simple switches, because they lie on only one line spanned by points of $P$ and the line at infinity.
The obstacle in the replacement procedure is, that we cannot replace the points on the line at infinity since they are determined by $P$.


To conclude, we have to adapt the replacement procedure for the
order type $O_A$, resp. the allowable sequence, of the proof of
Theorem~\ref{thm:nonSimple} to an order type $O_A'$ with the following properties.
\begin{itemize}
	\item The points $O_A\setminus\ell_\infty$ are in general position.
	\item For each pair of points $p,q\in P$ there is
	a point in~$O$ that is collinear with $p$ and $q$.
	\item No point on $\ell_\infty$ is collinear with two pairs of points of $O\setminus\ell_\infty$.
	\item A realization of $O_A'$ leads to a representation of $O_A$ and vice versa.
\end{itemize}
We call a point of $O$ \emph{active} if it lies on at least two lines induced by $P$. 

\begin{observation}\label{obs:inftyReplacement}
	Let $p_1,\dots,p_k$ be the constructability sequence restricted on
	$P$ and let $O_i$ be the active vertices of $O$ with respect to the vertices
	$p_1,\dots,p_i$.
	Each point of $p_i$ lies on at most two lines induced by
	$S_{i-1}:=\{o_1,\dots,o_{i-1}\}\cup P_{i-1}$.
\end{observation}

We call an allowable sequence with an order of the points satisfying
the observation above a \emph{constructible} allowable sequence.
To conclude the proof of Theorem~\ref{thm:simple} it is sufficient to show the following lemma.

\begin{lemma}\label{lem:constructible}
	Let $A$ be a constructible allowable sequence.
	There exists a simple allowable sequence $A'$ that is realizable if and only if $A$ is realizable.
\end{lemma}
\begin{proof}
	Let $p_1,\dots,p_n$ a constructability sequence of $A$.
	The goal is to construct a simple allowable sequence $A'$
	that is realizable if and only if $A$ is realizable.
	We achieve by replacing the point $p_n$ by several points $p_n^1,\dots,p_n^c$,
	such that $(p_1,\dots,p_{n-1},p_n^1,\dots,p_n^c)$ is a
	constructability order of a new allowable sequence $A^{n-1}$.
	In addition, the points $p_n^{1},\dots,p_n^{c}$ are not collinear with an active point of the new allowable sequence.
	This last property allows us to continue with the replacement procedure with the point $p_{n-1}$.
	We denote by $S_k$ the points $P$ of $O_{A^k}$ and its active points of $O$. 
	The number of points $c$ we replace $p_n$ by depends on the following cases.
	\begin{enumerate}
		\item\label{itm:noline} $p_n$ does not lie on a line spanned by $S_{k-1}$.
		\item\label{itm:oneline} $p_n$ lies on exactly one lines spanned by $S_{k-1}$.
		\item\label{itm:twoline} $p_n$ lies on exactly two lines spanned by $S_{k-1}$.
	\end{enumerate}
	In Case~\ref{itm:noline} there is nothing to change. We just set $A^{k-1}$ to $A^k$.
	We proceed with Case~\ref{itm:twoline} since Case~\ref{itm:oneline} can be solved
	exactly the same way, but the necessity of some details become clearer in Case~\ref{itm:twoline}.
	We replace the point $p_k$ by four points $p_k^1,\dots,p_k^4$.
	Those points are placed almost on the vertices of a very small
	parallelogram around $p_n$ as shown in Figure~\ref{fig:slopes_quadrant}.
	
	\begin{figure}[ht]
		\centering
		\includegraphics[width=.9\textwidth]{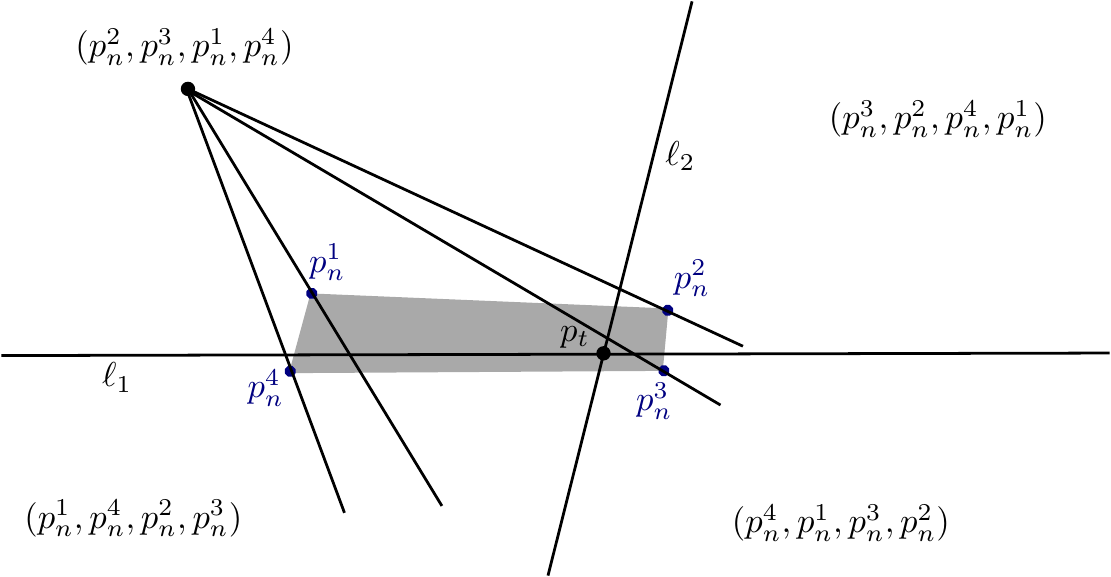}
		\caption{\label{fig:slopes_quadrant} Placing the points $p_k^i$
			around $p_k$.
			The order of switches between new points and a point in a quadrant is given by the tuples.}
	\end{figure}
	
	We can replace each occurrence of $p_k$ in a simple switch $(p,p_k)$
	by the four consecutive switches
	$(p,p_n^i),(p,p_n^j),(p,p_n^k),(p,p_n^l)$, where the exact order
	of those switches depends on the quadrant $p$ lies in as shown in
	Figure~\ref{fig:slopes_quadrant}.
	To determine how to replace the non-simple switches we make the
	following assumptions  on the position of the points $p_n^i$.
	\begin{itemize}
		\item The points $p_n^1$ and $p_n^2$ lie ``much'' closer to $\ell_1$
		than $p_n^3$ and $p_n^4$.
		This has the effect that the angle between $\ell_1$ and a line
		through  $p_n^1$ or $p_n^2$ and a point on $\ell_1$ is smaller
		than the angle between $\ell_1$ and $p_n^3$ or $p_n^4$ and any
		point on $\ell_1$ as shown in Figure~\ref{fig:switch_replacement}.
		\begin{figure}[ht]
			\centering
			\includegraphics[width=.9\textwidth]{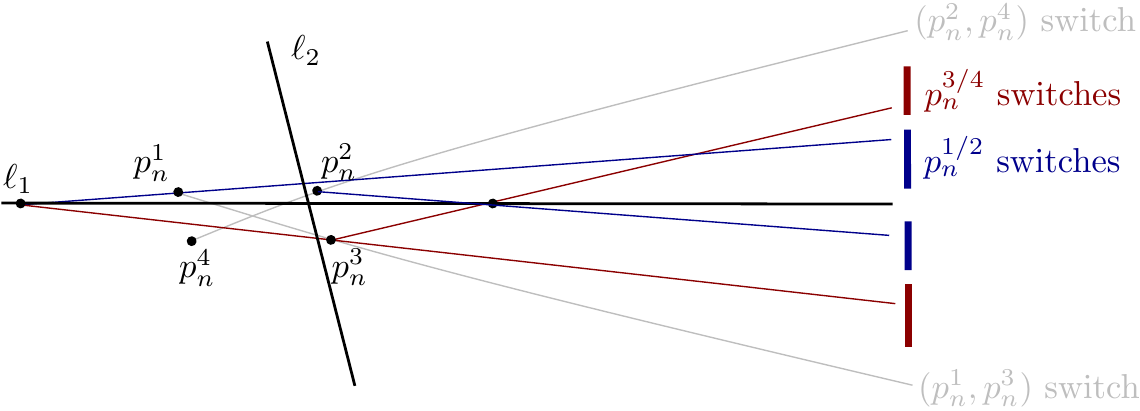}
			\caption{\label{fig:switch_replacement} The relative position of
				the switches through the points
				$o_n^i$ and a point on $\ell_1$.}
		\end{figure}
		Similarly, the points $p_n^2$ and $p_n^3$ lie closer to $\ell_2$ than
		$p_n^1$ and $p_n^4$,
		such that the order of the new switches for points on $\ell_2$ can be determined.
		\item The parallelogram is much larger in $\ell_1$ direction, such that the switches of
		the diagonals of the parallelogram are almost parallel to $\ell_1$ as indicated in Figure~\ref{fig:switch_replacement}.
		\item Finally, we perturb the points $p_n^2$ and $p_n^4$ slightly into the parallelogram,
		such that the switches given by the edges of the parallelogram are the closest switches to
		the switches on $\ell_1$ or $\ell_2$.
	\end{itemize}
	We now proceed with the case~\ref{itm:oneline}, i.e., $p_n$ lies on exactly one line spanned by the
	points $p_1,\dots,p_{n-1}$. We observe that we can just repeat the
	construction for two lines by picking an arbitrary line through $p_n$,
	and add a second point $p$ on this line.
	Then we apply the same replacement procedure as above and remove the
	added point $p$ as well as all points of the allowable sequence constructed because of $p$ again.
\end{proof}

\begin{remark}
	The results so far show universality of the realization space of allowable sequences even
	if the realization space of the induced order type is non-empty. We can also achieve universality
	of the allowable sequence and the order type simultaneously by
	considering an order type that also contains the points on the line at infinity
	(optionally: only the variable points on $\ell$). By a projective transformation
	we perturb $\ell$ away from the line at infinity without changing the
	allowable sequence
	of the points not on $\ell$. We can obtain the new allowable sequence by considering the
	rotation system of points $O$ around a point $p$ on $\ell$ and adding the switches through
	$p$ in an interval close to the old position of $p$ on the line at infinity as
	shown in Figure~\ref{fig:simple_perturb}.
	\begin{figure}[ht]
		\centering
		\includegraphics[width=.9\textwidth]{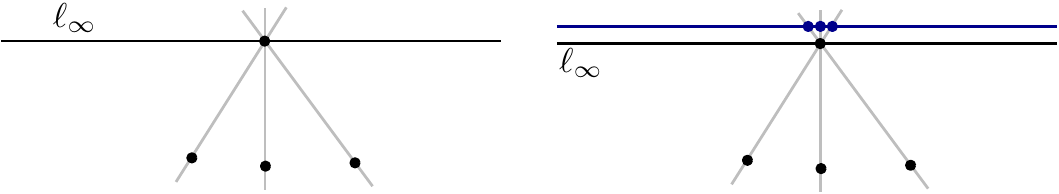}
		\caption{\label{fig:simple_perturb} Creating a new allowable sequence by perturbing the old allowable sequence away from the line at infinity.}
	\end{figure}
	
	The resulting allowable sequence is also constructible since none of the
	points of the allowable 
	sequence is active and the underlying order type is essentially the order
	type from Shor's proof
	(except for a slightly different placement of the gadgets) which is constructible.
\end{remark}

\section{Finite convex geometries}\label{sec:cg}
In this section we use the \ER-hardness result for the realizability of allowable sequences to
show that deciding the realizability of an abstract convex geometry in the plane is \ER-hard.
Afterwards we show that the convex geometries constructed in our reduction are realizable in an arbitrary dimension
if and only if they are realizable in the plane, which proves the following theorem.
\begin{theorem}\label{thm:convexG}
The realizability problem for abstract convex geometries is \ER-complete.
\end{theorem}
\subsection{Preliminaries}

We first recall that an (abstract) order type carries more information than the convex geometry.

\begin{observation}[\cite{Adaricheva10a}]\label{obs:AOtoCG}
  An abstract order type uniquely determines an abstract convex geometry.
\end{observation}
The observation above follows from the fact that for four points $a,b,c,d$ the point $d$ lies in the convex hull of $a,b,c$ if and only if the alphabetically ordered triples except $(a,c,d)$ are ordered clockwise, see Figure~\ref{fig:conv_orientation}.
\begin{figure}[!h]
	\centering
	\includegraphics{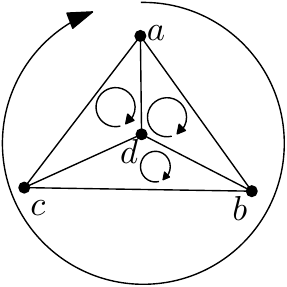}
	\caption{\label{fig:conv_orientation} Order types determine the convex geometry.}
\end{figure}
If an abstract convex geometry is realizable there exists an order type that encodes this convex geometry.
We show that the realizability problem for abstract convex geometries is \ER-complete, even when restricted to the abstract geometries that are defined by abstract order types.
As a consequence we can use an abstract order type as a polynomial encoding of the set system of possibly exponential size in the ground set that defines the convex geometry. We proceed to define some abstract order types which are the building blocks for our reduction.

We define $D_{2k}$ as the (abstract) order type of the following point set.
Consider $k$ lines that intersect in the common center point of two circles of almost the same radius.
Place a point on each intersection point of a circle with the $k$ lines.
We denote the points on the outer circle by $r_i$ and the inner ones by $r'_{i}$ in counterclockwise order.
The difference of the radii is small enough, such that $r'_i$ lies on the convex hull (of the set) when we remove $r_i$.
We call $D_{2k}$ a \emph{double ring} and denote the induced convex geometry by $(D_{2k},\CCC_{2k})$. We use the following convention, $r_i = r_{(i\mod k)+1}$.  

Note that the only collinearities in $D_{2k}$ appear among the points $r_i,r'_i,r_{i+k}$ and $r'_{i+k}$.
This means we can slightly perturb the lines in the construction of $D_{2k}$,
such that they do not intersect the center but form an arbitrary line arrangement in the neighborhood
of the center, and still obtain the same abstract order type.
We will use this fact and a ``unique representation'' we obtain from Lemma~\ref{lem:doublering},
which we will prove in the rest of this subsection, to fix an allowable sequence of a point set with a double ring. 


\begin{figure}[!h]
\centering
 \includegraphics[width=.9\textwidth]{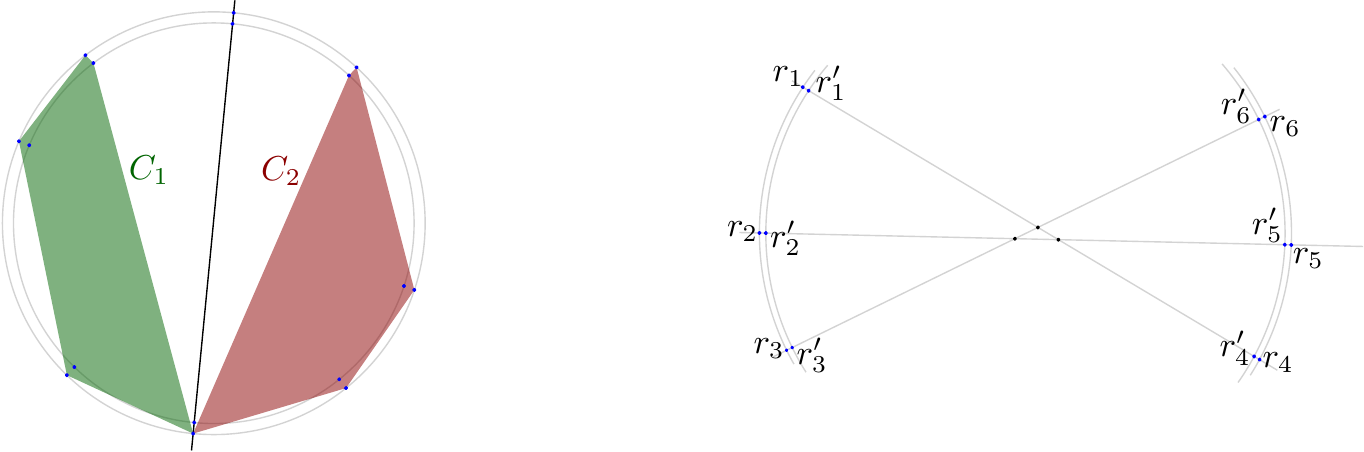}
  \caption{\label{fig:convexgeometry} Left: A double ring $D_{8}$ with the two
  maximal convex sets of Proposition~\ref{prop:maxConvex}. Right: Using the
  double ring to fix an allowable sequence.}
\end{figure}

\begin{proposition} \label{prop:maxConvex}
  Let $C$ be a convex set of  $(D_{2k},\CCC_{2k})$ such that $r_i\in C$, $r'_i\notin C$ for some $i$, and $C$
  is inclusion wise maximal for this property.
  Then we have $$C=  \{r_i\} \cup \{ r_{i+1},r'_{i+1},\ldots, r_{i+(k-1)},r'_{i+(k-1)} \}$$
  or $$C=\{r_i \}\cup \{ r_{i-1},r'_{i-1},\ldots, r_{i-(k-1)},r'_{i-(k-1)} \}.$$
\end{proposition}
The two possible convex sets (containing $r_i$ but not $r'_i$) are shown in Figure~\ref{fig:convexgeometry} left.

\begin{proof}
  Suppose $C$ is an inclusion wise maximal convex set with $r_i\in C$ and $r'_i\not\in C$. Then $C_i$ does not contain a point $r_j$ or $r'_j$ with $j\in \{j+1,\dots,i+k\}$ and a point $r_k$ or $r'_k$ with $k\in \{i-1,\dots,i-k\}$.
  This gives a bipartition of the points, not two points of the same class are in $C$.
  On the other hand, the union of one class and $\{r_i\}$ is a convex set, namely one of the sets given above.
\end{proof}

From this proof we obtain the following observation.

\begin{observation}\label{obs:sides}
  In each representation of $(D_{2k},\CCC_{2k})$ each side of the line $\ell(r_i,r'_i)$ contains exactly $k-1$ pairs of points $r_j,r'_j$.
  The points on one side form one of the maximal convex set containing $r_i$ and not $r_i'$.
\end{observation}
  
\begin{lemma}\label{lem:doublering}
  In each representation of $(D_{2k},\CCC_{2k})$ the order of the points on the convex
  hull is (up to a reflection) $(r_1,\ldots,r_{2k})$.
\end{lemma}

\begin{proof}
  Suppose there is a representation of the convex geometry induced by $D_{2k}$ with different order of points
  on the convex hull.
  Then there are $r_i$ and $r_{i+1}$ that are not adjacent on the convex hull.
  Thus there is a point $r_x$ that lies between $r_i$ and $r_{i+1}$ on the shorter path $p$ along the convex hull.
  Now $r_i$ and $r_{i+1}$ lie on different sides of the line $\ell(r_x,r'_x)$,
  since each side contains at most $k-1$ pairs of $r_l$ and $r'_l$ since $|p|<k$ by Observation~\ref{obs:sides}.
  The points on each side of $\ell(r_x,r'_x)$ together with $r_x$ form  maximal convex sets containing $r_x$ and not $r_x'$.
  According to Proposition~\ref{prop:maxConvex} there is no point $r_x$, such that $r_i$ and $r_{i+1}$
  are contained in the different maximal convex sets that do not contain $r'_x$, a contradiction. 
\end{proof}

\subsection{The reduction}

To prove Theorem~\ref{thm:convexG} we first show the 2-dimensional version.
\begin{theorem}\label{thm:convexG2d}
    The realizability of an abstract convex geometry in $\R^2$
    is \ER-complete.
\end{theorem}

\begin{proof}
We will reduce from realizability of allowable sequences.
Therefore, we define an order type $O_A$ from an allowable sequence $A$ in the following way.
Let $P_A$ be the order type that is induced by the allowable sequence $A$.
To define $O_A$ we add a double ring $D_{2k}$ to $P_A$, where $k$ is the number of switches in the allowable
sequence (e.g., $k=\binom{n}{2}$ if $A$ is simple).
This is done such that the double ring forms the two outer layers of the point set.
Furthermore, the line spanned by opposite pairs $p_i,p_{i+k}$ of $D$ contains the points of $P_A$ that are reversed in the $i$-th switch of $A$.
To complete the definition of $O_A$ we have to define the orientation of triples containing two points of $P$ and one point of $D$ and vice versa.
A triple $(p_h,p_i,r_j)$ with $j\leq k$ is oriented clockwise if $p_h$ appears before $p_j$ in the
$j$-th permutation of $A$, and counterclockwise otherwise.
The orientation of $(p_h,p_i,r_{j+k})$ is the reverse of $(p_h,p_i,r_j)$.
A triple $(r_h,r_i,p_j)$ is oriented clockwise if $h-i<k$ (cyclically) and counterclockwise if $h-i>k$.
Each triple $(r_h,r_{h+k},p_i)$ (with $h<k$) is oriented clockwise if $p_i$ lies before the substring that is reversed in the $h$-th switch of $A$ (and counterclockwise otherwise).
The triple including $r'_h$ has the same orientation as the triple where $r'_h$ is replaced by $r_h$. If $r_h$ and $r_h'$
are included in a triple then it has the same orientation as the triple where $r'_h$ is replaced by $r_{h+k}$.

By Lemma~\ref{lem:doublering} we know that each realization of the abstract geometry $C_A$ induced by $A$
contains a copy of the double ring $D$ with the same order of vertices on the convex hull.
The line spanned by $r_k$ and $r_{k+n}$ contains the points of $P$ that are reversed in switch $k$.
All those lines spanned by opposite pairs of the double ring already intersect in the interior of the convex hull,
thus the lines intersect the line at infinity in the same order as the convex hull of the double ring.
This shows that the order of switches is the same as in $A$, hence we obtain a realization of $A$ by considering
the sub-realization of $P_A$ in $R_{O_A}$ (a realization of $O_A$). On the other hand, $O_A$ is realizable if $C_A$ is realizable as we will sketch in Lemma~\ref{lem:realOT}.

Since all proofs in this section are constructive and can be implemented in polynomial time,
we have reduced allowable sequence realizability to convex geometry realizability which concludes the proof of Theorem~\ref{thm:convexG}.
\end{proof}

\begin{lemma}\label{lem:realOT}
  The order type $O_A$ is realizable if $A$ is realizable, where $O_A$ and $A$ are defined as in the proof of Theorem~\ref{thm:convexG2d}.
\end{lemma}

\begin{proof}
  Consider a realization $R_A$ of $A$. We can realize $O_A$ by placing the points of the double ring $D$ on the intersection
  points of the lines spanned by $P$ with two ``very large'' circles of almost the same radius that contain the realization
  $R_A$ ``close'' to their center point.
\end{proof}

Thus it remains to show that any realization of the convex geometry $C_A$ induced by $O_A$ gives a realization of $A$. Increase the dimension to obtain the following proof. 

\begin{proof}[Proof of Theorem~\ref{thm:convexG}]
  To prove Theorem~\ref{thm:convexG} it is sufficient to show that in each representation of the abstract
  convex geometry $C_A$ in $\R^d$ all points lie in one plane.
  Therefore, notice that there is a triangle of the double ring that contains all points of $P_A$.
  This implies that all points of $P_A$ are coplanar in each representation of $C_A$ in $\R^d$.
  All points of the double ring $D$ lie on a line spanned by points of $P$.
  Consequently, all points of $C_A$ are coplanar and $C_A$ is realizable in $\R^2$ if and only if it is realizable in $\R^d$.
\end{proof}

\section{Polygon visibility graphs with holes}\label{sec:pv}

In this section we show that the recognition of visibility graphs of polygons with holes is \ER-complete.
First we show that the problem is \ER-complete if we know the cycles of the
graph that bound the holes. In a second step we show that this condition can be
dropped.

The general idea is similar to the reduction in the last section. We use the
outer face cycle of the polygon to fix the allowable sequence, the holes
represent the points realizing the allowable sequence.

In this section we denote the open straight-line segment between two vertices of the polygon as \emph{sightline}. We say a sightline is blocked if the line segments intersects the polygon.

\subsection{Reduction with given boundary cycles}

\begin{theorem}
  The recognition problem for visibility graphs of polygons with holes with given boundary cycles is \ER-complete.
\end{theorem}

  We reduce the realizability of simple allowable sequence to the recognition of visibility graphs of polygons with holes.
  So, given an allowable sequence of $n$ points we construct a graph $G=(V,E)$
  along with a partition of $V=\{v_1,\dots,v_{|V|}\}$ into $n+1$ sets.
  Each of the sets gives the vertices of one boundary cycle.
  The order along the cycle respects the order of the indices of the vertices.
  The size of $G$ is polynomial in $n$.

  The general idea is the following.
  The holes reserve an area for the points of the point set.
  We use the outer face cycle of the polygon to mark the position
  of the switches of the allowable sequence, i.e., we construct a line through each pair of points and
  where this line hits a circle of very large radius we place a point.
  On both sides of each point $s$ on the circle we place two more points, which bound an interval containing $I_s$.
  The intervals are disjoint.
  
\begin{figure}[!h]
\centering
\includegraphics[width=.9\textwidth]{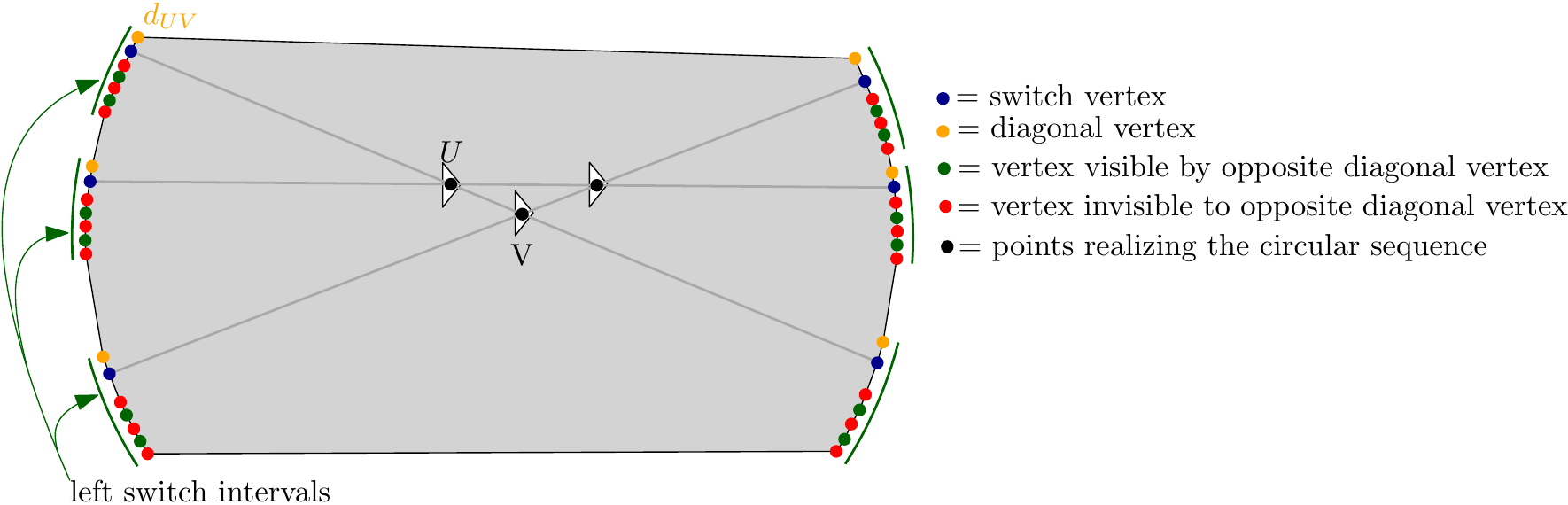}
 \caption{\label{fig:constr_PVG} An overview of the construction of a visibility graph from an allowable sequence.}
\end{figure}
  The \emph{diagonals} connecting the endpoints of two opposite intervals (intervals of the same switch $s_{a,b}$)
  separate the points $a,b$. Now we replace the points of the point set by holes of the polygon.
  The holes are still separated by the diagonals. This has the effect that we can pick one representative point
  per hole. These representatives still have the original allowable sequence.
  This is due to the fact that the line though each pair of points still lies
  in the interval of the correct switch since it is framed by the diagonals.
  Thus it remains to define a graph $G_A$ and a partition into boundary cycles from the
  allowable sequence $A$ that has the described \emph{unique representation} as polygon visibility graph.

  The outer face cycle consists of $2\binom{n}{2}(2n+1)$ vertices, $2\binom{n}{2}$ intervals (two opposite intervals per switch) of $(2n+1)$ vertices.
  The first $\binom{n}{2}$ of the intervals we call left intervals, the
  remaining ones right intervals as shown in
  Figure~\ref{fig:constr_PVG}.
  We describe the missing edges of the outer face cycle.
  The second highest point of each interval is the switch vertex and is marked blue in Figure~\ref{fig:hiding}. It has an edge with every vertex of the outer face cycle except for its opposite switch vertex.
  The highest vertex of each interval shares no edge with every second of the lower $2n-1$ opposite vertices, starting with the lowest (i.e., $n$ edges are missing).
  Those are all missing edges of the outer boundary cycle.
  
  Each of the $n$ holes consists of three vertices that form a triangle. 
  One vertex we call the \emph{obtuse vertex} since it has an obtuse angle when we construct a representation from an
  allowable sequence\footnote{This angle is not obtuse in every representation.
  It can be made non-obtuse by a linear transformation of the representation.}.  
  The allowable sequence (with our definition of left and right switches) gives a total ordering of the holes from left to right.
  The visibility among the hole vertices will be defined as follows.
  For two holes all three vertices of the left hole see the non-obtuse vertices of the right hole.
  All points of the right intervals see all vertices of the holes, except for the switch vertices which do not see any vertex of the left hole of its switch.
  Almost symmetrically, all vertices of the left intervals see all non-obtuse vertices of the holes,
  except for the switch vertices that do not see any vertex of the right hole of its switch.
  
  \begin{lemma}\label{lem:polyvisRealizable}
    Let $A$ be a realizable allowable sequence. Then $G_A$ is realizable as polygon visibility graph.  
  \end{lemma}
 
    \begin{proof}
    	Consider a point set $P_A$ realizing $A$.
    	From this realization we construct a polygon visibility representation of $G_A$.
    	Let us first construct the outer polygon boundary.
    	Therefore we introduce a bounding stripe for each switch, i.e., an area containing the point set,
    	which is only bounded by two parallel lines parallel to the line of the switch.
    	Furthermore, we assume that the point set is strictly contained in the top half of the stripe.
    	We consider a large circle around the point set.
    	On the upper endpoint of each interval we place the diagonal point.
    	If we choose the circle large enough, then we can assume that
    	\begin{itemize}
    		\item the intervals on the circle are disjoint,
    		\item the lines through the diagonal points and a point of
    		$P_A$ intersect the opposite interval below the switch vertex.
    	\end{itemize}
    	We mark the lower points of each interval on the intersection points of the lines through the opposite diagonal vertex and the points of $P_A$.
    	Between each two of those lower points we add the remaining $n-1$ vertices.
    	Note that this construction already leads to a polygon visibility with points instead of triangles as holes.
    	We replace the holes by small triangles each of the same shape and size: a vertical segment gives the position of two of the vertices,
    	the third vertex is placed close to the midpoint of the segment, perturbed by a small $\epsilon$ to the right.
    	We pick one point $p$ of the interior of the triangle and put translates of the triangle onto the points of $P_A$, such that $p$ lies on the point of $P_A$
    	We slightly increase the size of the triangles, such that no additional visibility is blocked.
    \end{proof}
    
  In the following we use the expression \emph{a vertex sees a hole} and \emph{a vertex does not see a hole}
  instead of \emph{a vertex sees a vertex of the hole} and \emph{a vertex does not see any vertex of the hole}.

  \begin{lemma}\label{lem:visBlock}
    Let $R_A$ be an arbitrary polygon visibility representation of $G_A$. The visibility of the switch vertex $s_{UV}$ is blocked from $V$ entirely by $U$ in $R_A$. 
  \end{lemma}
   See Figure~\ref{fig:blockHole}.
   
 \begin{proof}
 	Each vertex of the hole $V$ does not see exactly $n-1$ vertices of the outer boundary cycle.
 	There is a set of $3(n-1)$ sightlines that have to be blocked,
 	namely the sightlines from the vertices of $V$ to the switch vertices of a switch of $V$.
 	One of the three sightlines of $V$ to a single switch vertex may be blocked by $V$ itself,
 	but two of the sightlines -- the sightlines to the extreme vertices of the hole -- require another hole as a blocker.
 	Those sightlines can be partitioned into two sets (purple and red in Figure~\ref{fig:hiding} left);
 	the sightlines between the clockwise extreme and counterclockwise extreme vertices seen from a switch vertex.
 	We observe that the sightlines of neither set are crossing.
 	Which of the vertices of the hole are the extreme vertices depends on the region
 	bounded by the extended sides of the polygon (Figure~\ref{fig:hiding} right) the switch vertex lies in.
 	Two non-crossing sightlines that are not starting from the same switch vertex are separated by a sightline that is not blocked
 	(the points in the intervals around the switch vertices).
 	Thus each other hole can block at most two sightline, and those sightlines are either crossing or incident
 	to the same switch vertex. We show that the latter is the case.
 	
 	Now from each switch vertex each pair of holes is separated by a sightline as shown in the proof of Lemma~\ref{lem:diagonals}.
 	This is not the case for a switch vertex that lies on the boundary between the ends of the blocked crossing sightline, see Figure~\ref{fig:hiding} right.
 	\begin{figure}[!h]
 		\centering
 		\includegraphics[width=.9\textwidth]{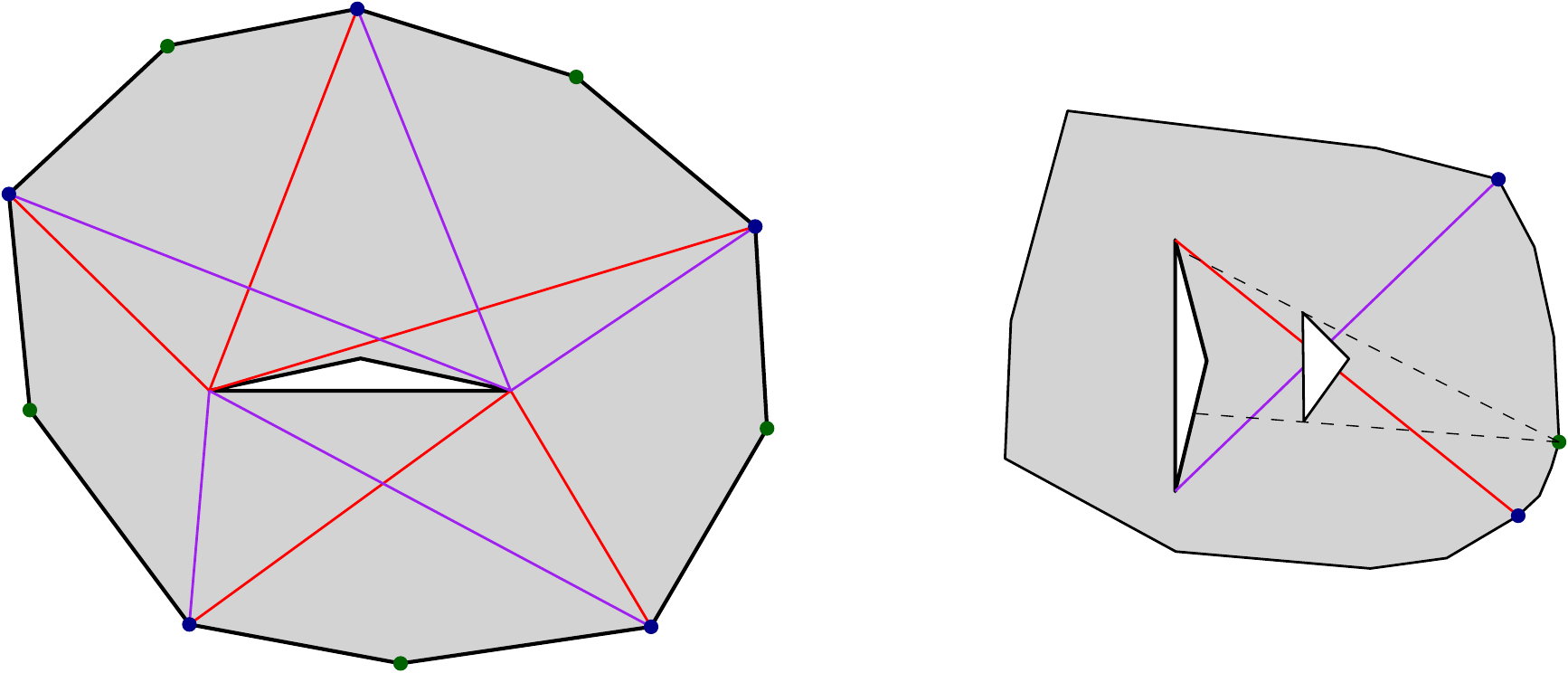}
 		\caption{\label{fig:hiding} Left: sightlines to the extreme vertices. Sightlines of the same type are non-crossing. Right: Two holes that cannot be separated by a non-crossing sightline if two crossing sightlines are blocked by the same hole. }
 	\end{figure}
 \end{proof}

\begin{figure}[!h]
\centering
 \includegraphics[width=.5\textwidth]{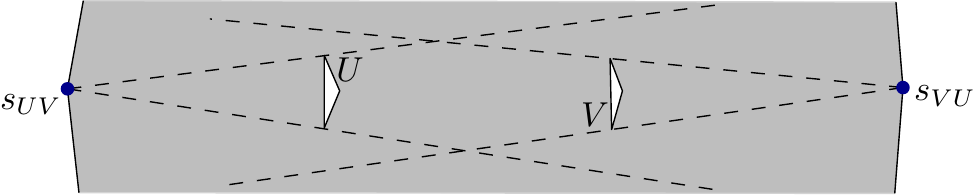}
 \caption{\label{fig:blockHole} The $s_{UV}$ is blocked from $V$ entirely by $U$.}
\end{figure}
  \begin{observation}
    The line $\ell_{UV}$ that connects $s_{UV}$ and $s_{VU}$ intersects $U$ and $V$. 
  \end{observation}

  \begin{lemma}\label{lem:diagonals}
    There is a sightline line $D_{UV}$ connecting $d_{UV}$ and one of the vertices in the interval below $s_{VU}$ that lies above $U$ and below $V$. 
  \end{lemma}
\begin{proof}
    The point $d_{UV}$ does not see $n$ of the points below $s_{vu}$.
    The $n$ sightlines that have to be blocked are separated by sightlines that are not blocked, thus we require all $n$ holes as blockers.
    Consequently, there is one non-blocked sightline that separates $U$ and $V$.
    Since $D_{UV}$ intersects $\ell_{UV}$ from above as shown in Figure~\ref{fig:diagonal1}, we know that $U$ lies below and $V$ above $D_{UV}$.  
  \end{proof}
  
\begin{figure}[!h]
\centering
 \includegraphics[width=.5\textwidth]{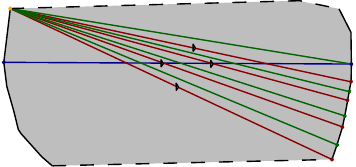}
  \caption{\label{fig:diagonal1} There is a diagonal sightline from $d_{UV}$ separating
  the holes $U$ and $V$.}
\end{figure}

\begin{lemma}
    Let $S$ be a point set that contains exactly one point of each hole.
    The allowable sequence of $S$ is $A$.
  \end{lemma}

  \begin{proof}
    The line $\ell_{UV}$, which is spanned by the representatives of hole $U$ and $V$, ends in the intervals
    around $s_{UV}$ and $s_{VU}$ which are defined by the points $D_{UV}$ and
    $D_{VU}$ end in, see Figure~\ref{fig:pockets_diagonals} left.
    Since those intervals around the switch vertices are disjoint (by Lemma~\ref{lem:diagonals}), we know that the
    order of the endpoints on the boundary of the polygon are as given by $A$.
  \end{proof}

\begin{figure}[!h]
\centering
 \includegraphics[width=.85\textwidth]{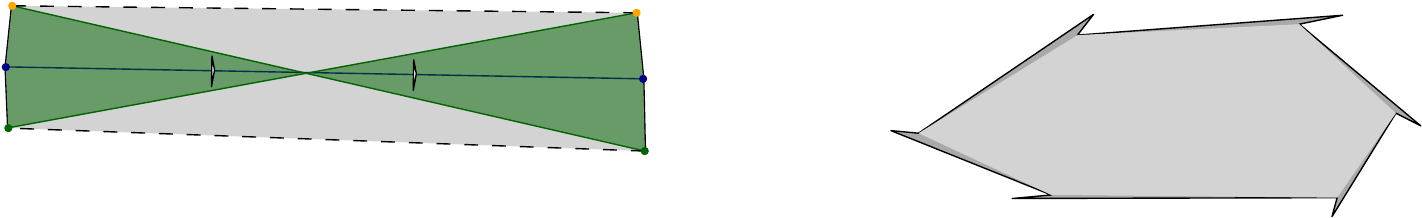}
 \caption{\label{fig:pockets_diagonals} Left: With each choice of representative point of a
  hole the spanned line lies in the green double wedge bounded by the diagonal
  sightlines. Right: Subdividing the edges on the outer boundary of a convex polygon.}
\end{figure}

\subsection{Fixing the boundary cycles}
  In this section we give a small modification to the reduction above to show that we do not
  need to fix the boundary cycles to obtain \ER-hardness.

  We modify the graph $G_A$ by subdividing each edge of the outer face cycle and call the graphs $G_A'$.

  \begin{lemma}
    The graph $G_A'$ is realizable as polygon visibility graph if and only if $G_A$ is realizable.
  \end{lemma}

  \begin{proof}
    Each subdivision vertex we added to $G_A'$ is adjacent to its two neighbors
    in each partition into boundary cycles,
    thus the face cycle of the outer boundary is as in $G_A$.
    This cycle is the outer boundary cycle since each triple of formerly consecutive
    outer boundary vertices can see each other, thus gives a convex angle.
    The only boundary cycle that can be formed of convex angles is the outer face cycle.

    The remaining $3n$ vertices have to form $n$ holes, because we need $n$ holes to block
    the visibilities of a diagonal vertex of an interval from the points in the opposite
    interval as shown in the proof of Lemma~\ref{lem:diagonals}.
    Thus each hole consists of exactly three vertices.
    The obtuse vertex of the rightmost vertex has degree two in the induced subgraph of
    the vertices that are not assigned into a cycle yet,
    thus forms a boundary cycle together with its neighbors.
    Iteratively, the obtuse vertex of the rightmost hole has degree
    two in the induced subgraph of vertices that is not assigned to a cycle yet,
    which gives exactly the cycle partition as constructed.
    We obtain a realization of $G_A'$ from a realization of $G_A$ by adding a subdivision vertex in
    a pocket next to a convex corner as shown in
    Figure~\ref{fig:pockets_diagonals} right.
    The uniqueness of the representation of $G_A'$ holds just as for $G_A$.
  \end{proof}

\subsection*{Acknowledgments}
 We thank Jean Cardinal and Stefan Felsner for fruitful discussions.

\bibliographystyle{plain}
\bibliography{alles}

\end{document}